\documentclass[]{article}
\usepackage{amsmath,amsfonts,amssymb}
\usepackage{amsthm}
\usepackage{mathtools}
\usepackage{cases}
\usepackage[mathfrak]{}
\usepackage[latin1]{inputenc}
\usepackage[T1]{fontenc}
\usepackage{verbatim}
\usepackage{graphicx}
\usepackage{float}
\usepackage{mathrsfs}
\usepackage [all]{xy}
\usepackage{color}
\usepackage{tikz}
\usepackage{comment}
\usepackage{hyperref}
\usepackage{hyphenat}
\usepackage{authblk}
\newtheorem{theorem}{Theorem}[section]% meant for sectionwise numbers
\newtheorem{lemma}[theorem]{Lemma}
\newtheorem{claim}[theorem]{Claim}
\newtheorem{proposition}[theorem]{Proposition}% 

\newtheorem{example}{Example}%
\newtheorem{remark}{Remark}%
\newtheorem{protocol}{Protocol}

\newtheorem{definition}{Definition}

\title{A generalization of Burmester-Desmedt GKE based on a non-abelian finite group action}
\author[1]{Daniel Camaz\'on-Portela \footnote{The author was partially supported by grant PID2022-138906NB-C21 funded by MICIU/AEI/ 10.13039/501100011033 and by ERDF A way of making Europe.}}
\author[2]{\'Alvaro Otero-S\'anchez \footnote{The second author was partially supported by Grant FPU2024 funded by MICIU/AEI /10.13039/501100011033 and FSE+}}
\author[3]{Juan Antonio L\'opez-Ramos \footnote{The third author was partially supported by grant PID2022-140934OB-I00 funded by by MICIU/AEI/ 10.13039/501100011033 and by ERDF A way of making Europe.}}
\affil[1, 2, 3]{Department of Mathematics, University of Almer\'ia, Carretera Sacramento, SN, Almer\'ia, 04120, Spain}
\date{}                     %% if you don't need date to appear
\begin{document}
  \maketitle

%\date{}

%\begin{document}
%\maketitle

\begin{abstract}
The advent of large-scale quantum computers implies that our existing public-key cryptography infrastructure has become insecure. That means that the privacy of many mobile applications involving dynamic peer groups, such as multicast messaging or pay-per-view, could be compromised. In this work we propose a generalization of the well known group key exchange protocol proposed by Burmester and Desmedt to the non-abelian case by the use of finite group actions and we prove that the presented protocol is secure in Katz and Yung's model.
\end{abstract}

\section{Introduction}\label{sec1}

Nowadays, many mobile applications involve dynamic peer groups. Some examples of these include audio/video conferencing, multicast messaging, pay-per-view, and, in general, distributed applications. In order to protect such applications and preserve privacy, secure communication channels that link all the parties of a group must be employed. For achieving privacy, one can establish such channels by using symmetric encryption (symmetric cryptosystem, private key cryptosystem) with a group session key.

The aim of any group key exchange (GKE) protocol is to enable secure communication over an untrusted network by deriving and distributing shared keys between two or more parties. We should distinguish between two types of key exchange (KE) protocols commonly known as key transport protocol (GKT) and key agreement protocol (GKA). Key transport protocol is a key establishment protocol in which one of the principal generates the key and this key then is transferred to all protocol users, whereas key agreement protocol is a key establishment protocol in which the session key is a function of inputs by all protocol users.

The first GKA protocol based on the two-party Diffie-Hellman key agreement protocol was proposed by Ingemarsson, Tang in \cite{WongTang82}. This was followed
by the GKA protocols of Koyama and Ohta \cite{Koyama87}, Blundo et al. \cite{Blundo93}, and Burmester and Desmedt (BD) \cite{BurmesterDesmedt05}. Since then a large number of research papers on GKA and on securing GKA protocols have been presented due mainly to the distributed and dynamic nature of GKA and to the security challenges that have to be resolved. Most of these protocols combine aspects of the Diffie-Hellman key agreement with other cryptographic primitives to support security. In a group key exchange, the number of rounds is a crucial measure for evaluating the efficiency and to obtain a constant-round GKE protocol is considered as a minimum desirable requirement. Traditionally, the Burmester and Desmedt (BD) protocol has been widely known from its simplicity and small round complexity, just two rounds. Subsequently, Just and Vaudenay generalized in \cite{JustVaudenay96} the BD construction in which any two-party KE can be used for obtaining GKE. However, their description was sketchy and a rigorous security proof was not presented before \cite{KatzMoti07} and \cite{Boyd03}.

However, we know that in the advent of large-scale quantum computers, both the factoring and discrete logarithm problems can be efficiently solved, meaning that our existing public-key cryptography infrastructure has become insecure. The security of classical cryptographic schemes such as RSA, and Diffie-Hellman are based on the difficulty of factoring large integers and of finding discrete logarithms in finite cyclic groups, respectively. A quantum computer is able to solve the aforementioned problems attacking the security of these cryptographic algorithms. More precisely, Shor's algorithm factors discrete logarithm problems and Grover's algorithm can improve brute force attacks by significantly reducing search spaces for private keys. As a result, many researchers are now interested in
cryptography that is secure in a post-quantum world.

In the post-quantum setting, there exist two variants BD-type GKE protocols: one from lattices \cite{Apon19} and the other one from isogenies \cite{Furukawa18}. Apon et al. proposed in \cite{Apon19} a lattice-based BD-type GKE from the Ring-LWE (R-LWE) assumption, in which the authors elaborately adjusted the original
security proof to their new post-quantum setting. However, since the underlying R-LWE assumption depends on the number of group members, $n$, the size of data also
gets large depending on $n$. In \cite{Furukawa18} Furukawa et al. proposed an isogeny-based BD type GKE protocol called SIBD. However, the security proof of SIBD (\cite[Theorem 4]{Furukawa18}) is imperfect, and several points remain unclear, for example, on how to simulate some public variables.

Group theory is a broad and rich theory that models the technical tools used for the design and analysis in this research.
Some of the candidates for post-quantum cryptography (PQC) have been known for years, while others are still emerging.
Group theory, and in particular non-abelian groups, offers a rich supply of complex and varied problems for cryptography (see e.g. \cite{Vasco24}); reciprocally, the study of cryptographic algorithms built from these problems has contributed results to computational group theory. Motivated by this fact, and inspired by the the works \cite{Maze07}, \cite{Lopez17} and \cite{Gnilke24}, where a suite of protocols for group key management based on general semigroup and group actions are provided, we propose a generalization of the original BD GKE protocol by means of the use of non-abelian finite group actions. In section \ref{sec2} we recall some concepts and come well known results about group actions. The protocol, that has two more extra rounds than the original BD protocol due to the non-commutative setting, is presented in section \ref{sec3}. However, we may assume that any two consecutive participants, $U_{i}$ and $U_{i+1}$, share a private key before running the protocol as in \cite[Protocol SA-GDH.2]{Ateniese00}, so Round $1$ could be omitted and the number of extra rounds would be reduced to one. Section \ref{sec4} is devoted to present the security model that will be used in section \ref{sec5} for its analysis, where we prove that the proposed protocol verifies correctness and achieves forward secrecy by adapting the proof for the original BD protocol presented in \cite{KatzMoti07}. Finally, in the last section we propose some group actions based on hard group theoretic problems that could be used to implement our GKE protocol.

\section{Preliminaries on finite group actions} \label{sec2}
In this section we recall some concepts and well known result about finite group actions.

\begin{definition}\label{DefGroupAction}
Let $(H,\odot)$ be a finite group and $X$ be a finite set. An action of $H$ on $X$ as a set is defined by a morphism $\phi: H\times X\rightarrow X$ verifying that:
\begin{enumerate}
\item $\phi(e_{H},x)=x$,
\item $\phi(h_{2},\phi(h_{1},x))=\phi(h_{2}\odot h_{1},x)$.
\end{enumerate}
\end{definition}

\begin{definition}
For each element $x\in X$, the set $H(x)=\left\{\phi(h,x)\vert\enspace h\in H\right\}$ is called the orbit of $x$, and its stabilizer is defined as $H_{x}\equiv\left\{h\in H\vert\enspace \phi(h,x)=x\right\}$.
\end{definition}

The next result known as the fundamental lemma establishes a relation between the cardinal of an orbit $H(x)$ and the order of $H$ and $H_{x}$.

\begin{lemma}\cite[Lemma 1.2.1]{Kerber99}
The mapping $H(x)\rightarrow H/H_{x}: \phi(h,x)\rightarrow hH_{x}$ is a bijection between the orbit $H(x)$ of $x$ and the set of left cosets $H/H_{x}$.
\end{lemma}
This result implies that the length of the orbit is the index of the stabilizer. In particular, since $H$ is finite, then $\left|H(x)\right|=\left|H\right|/\left|H_{x}\right|$.

To conclude this section, we introduce the concept of finite double action of a finite group $H$ on a finite set $X$.

\begin{definition}[Double group action]\cite[Definition 6.]{Habil19}
Let $(H,\odot)$ and $(J,\ast)$ be finite groups and $X$ a finite set. Let us suppose that $H$ and $J$ has an action on $X$ on the left $\phi: H\times X\rightarrow X$  and $\gamma: J\times X\rightarrow X$. For every $h\in H$, $j\in J$ and $x\in X$, if the interchange law $\phi(h,(\gamma(j,x)))=\gamma(j,\phi(h,x))$ holds, then we say that the pair ($H,J$) acts doubly on $X$ on the left by means of $\phi$ and $\gamma$. The action of ($H,J$) on $X$ defined by this way is called left double action. 
\end{definition}

\begin{example}
Let us consider two different actions of a finite group $(G,\cdot)$ on itself as follows:
\begin{align*}
\xymatrix@=0.5cm{
\phi: G\times G\ar[r] & G \\
(g_{1},g)\ar[r] & g_{1}\cdot g,}
%\end{equation*}
&
%\begin{equation*}
\xymatrix@=0.5cm{
\gamma: G\times G\ar[r] & G \\
(g_{2},g)\ar[r] & g\cdot g_{2}^{-1}.}
\end{align*}
Then, the group $(G,\cdot)$acts doubly on itself on the left by means of $\phi$ and $\gamma$.
\end{example}

\begin{remark}
Note that given an element $x\in X$ the notions of orbit and stabilizer can be extended to the case of double left actions.
\end{remark}

\section{A GKE protocol generalizing the BD protocol to the non-abelian case}\label{sec3}

We propose the following group key exchange protocol that generalizes the well-known Burmester-Desmedt Conference key Distribution System (see \cite{BurmesterDesmedt05}) to the non-commutative case.\\
Let $(G,\cdot)$ and $(H,\odot)$ be two finite groups. Then a center chooses an element $g\in G$ and publishes it. 

\begin{protocol}\label{Prot1}
Let $U_{1}, U_{2},...,U_{n}$ be a group of parties that want to generate a group key. \\
$\left[\textbf{Round 1}\right]$ Users $U_{i}$ and $U_{i+1 mod(n)}$, $i=1,...,n$ perform a key exchange and they get a shared private key $c_{i}\in H$. \\
$\left[\textbf{Round 2}\right]$ Each party $U_{i}$, $i=1,...,n$ selects an element $h_{i}\in H$ and sends $v_{i}=\phi(h_{i},g)$ to $i+1 mod(n)$ and $i-1 mod(n)$. \\
$\left[\textbf{Round 3}\right]$ Each party $U_{i}$, $i=1,...,n$ sends $w_{i}=\phi(c_{i-1}\odot h_{i},v_{i-1 mod(n)})$ to $i-1 mod(n)$. \\
$\left[\textbf{Round 4}\right]$ Each party $U_{i}$, $i=1,...,n$ computes $Y_{i}=\phi(c_{i}^{-1},w_{i+1})=\phi(h_{i+1 mod(n)}\odot h_{i},g)$, $X_{i}=\phi(h_{i},v_{i-1})=\phi(h_{i}\odot h_{i-1 mod(n)},g)$, and broadcast $Z_{i}=X_{i}^{-1}\cdot Y_{i}$. \\
$\left[\textbf{Key computation}\right]$ Each party $U_{i}$, $i=1,...,n$ computes the key
\begin{equation}
K_{i}=A_{\sigma^{i-1}(1)}^{i}\cdot A_{\sigma^{i-1}(2)}^{i}\cdot A_{\sigma^{i-1}(3)}^{i}\cdots A_{\sigma^{i-1}(n)}^{i},
\end{equation}
where
\begin{align*}
& A_{1}^{i}\equiv X_{i}=\phi(h_{i}\odot h_{i-1 mod(n)},g), \\
& A_{2}^{i}\equiv A_{1}^{i}\cdot Z_{i}=\phi(h_{i}\odot h_{i-1 mod(n)},g)\cdot\phi(h_{i}\odot h_{i-1 mod(n)},g)^{-1}\cdot Y_{i}, \\
& A_{3}^{i}\equiv A_{2}^{i}\cdot Z_{i+1 mod(n)}=\phi(h_{i+1 mod(n)}\odot h_{i},g)\cdot\phi(h_{i+1 mod(n)}\odot h_{i},g)^{-1}\cdot Y_{i+1 mod(n)}, \\
& ......, \\
& A_{n}^{i}\equiv A_{n-1}^{i}\cdot Z_{i+n-2 mod(n)}, \\
& =\phi(h_{i+n-2 mod(n)}\odot h_{i+n-3 mod(n)})\cdot\phi(h_{i+n-2 mod(n)}\odot h_{i+n-3 mod(n)},g)^{-1}\cdot Y_{i+n-2 mod(n)},
\end{align*}
and $\sigma\in S^{n}$ is the permutation $\sigma=(1,n,n-1,\ldots,2)$.
\end{protocol}

Note that in the commutative case, we can rid the protocol of Rounds $1$ and $3$ as well as of the permutation $\sigma$. As a consequence, in the commutative setting, it is enough for $(H,\odot)$ to have a semigroup structure.

\begin{example}[Commutative case]
Let $G=\mathbb{Z}_{p}$, $g$ an element of $G$, with $q=ord(g)$, and $H=\mathbb{Z}_{q}$. Then we can define a semigroup action $\phi: H\times G\rightarrow G$, where $\phi(h_{i},g)=g^{h_{i}}$. By running Protocol \ref{Prot1} we get $v_{i}=g^{h_{i}}$ and $Z_{i}=g^{h_{i+1}\odot h_{i}}\cdot (g^{h_{i}\odot h_{i-1}})^{-1}$, so the resulting group key is
\begin{equation*}
sk=g^{h_{1}\odot h_{n}}\cdot g^{h_{2}\odot h_{1}}\cdots g^{h_{n}\odot h_{n-1}},
\end{equation*}
and the original Burmester-Desmedt GKE protocol is obtained.
\end{example}

\section{The security model}\label{sec4}

We will use the security model developed from \cite{KatzMoti07} in order to formalize the security analysis that we are introducing in the next section. To this end we fix the notation and definitions necessary to such a formalization.

The (potential) participants in the protocol are modeled as probabilistic polynomial time (ppt) Turing machines in the finite set $\mathcal{U}=\{\mathcal{U}_1,\dots ,\mathcal{U}_n\}$ and every participant $\mathcal{U}_i$, $i=1, \dots ,n$ in the set $\mathcal{U}$ can run a polynomial amount of protocol instances in parallel. \\
The instance of participant $\mathcal{U}_i$ will be denoted as $\Pi_{U}^{i}$ and it is assigned with the following variables:
\begin{itemize}
\item $pid_{U}^{i}$ : stores the identities of the parties user $\mathcal{U}_i$ (including $\mathcal{U}_i$) aims at establishing a session key with variables assigned. 

\item $used_{U}^{i}$ : is a variable that indicates whether this instance is taking part in a protocol run.

\item $sk_{U}^{i}$ : is a variable that is initialized with a distinguished NULL value and will store the session key. 

\item $sid_{U}^{i}$ : is a variable that stores a non-secret session identifier to the session key stored in $sk_{U}^{i}$. 

\item $acc_{U}^{i}$ : is a variable that indicates whether the session key in $sk_{U}^{i}$ was accepted or not. 

\item $term_{U}^{i}$ : is a variable that indicates whether the protocol execution has finished.
\end{itemize}

Concerning the adversary's capabilities, since we are just considering passive adversaries, these are reduced to the following oracles:
\begin{itemize}
\item Execute($\mathcal{U}_1,i_1,\dots,\mathcal{U}_n,i_n$): in case the instances $\left\{\Pi_{U_{j}}^{i_{j}}\right\}$ have not yet been used, this oracle will return a transcript of a complete execution of the protocol among the specified instances.
\item Test($\mathcal{U},i$): this query is allowed only once, at any time during the adversary's execution. A random bit $b$ is generated; if $b=1$ the adversary is given $sk_{U}^{i}$, and if $b=0$ the adversary is given a random session key.
\end{itemize}

Following \cite{KatzMoti07} we recall the definitions of correctness and forward secrecy for a GKE protocol. 

\begin{definition}
A protocol is said to satisfy the correctness condition if for all $\mathcal{U}, \mathcal{U}^{'},i,j$ such that $sid_{U}^{i}=sid_{U^{'}}^{j}$, $pid_{U}^{i}=pid_{U^{'}}^{j}$, and $acc_{U}^{i}=acc_{U^{'}}^{j}=TRUE$, then $sk_{U}^{i}=sk_{U^{'}}^{j}\neq NULL$. 
\end{definition}

Regarding the security, we say that the event $Succ$ occurs if the adversary $\mathcal{A}$ queries the $Test$ oracle on an instance $\Pi_{U}^{i}$ for which $acc_{U}^{i}=TRUE$, and $\mathcal{A}$ correctly guesses the bit $b$ used by the $Test$ oracle in answering this query.

\begin{definition}
The advantage of $\mathcal{A}$ in attacking a protocol $P$ is defined as
\begin{equation}
Adv_{\mathcal{A},P}(k)\equiv\left|2\cdot Pr\left[Succ\right]-1\right|.
\end{equation}
We say that the protocol $P$ is a secure group key exchange ($KE$) protocol if it is secure against a passive adversary, that is, for any PPT passive adversary $\mathcal{A}$ it is the case that $Adv_{\mathcal{A},P}(k)$ is negligible. \\
To enable a concrete security analysis, we define $Adv_{P}^{KE-fs}(t,q_{ex})$ to be the maximum advantage of any passive adversary attacking $P$, running in time $t$, and making $q_{ex}$ calls to the $Execute$ oracle. We say that the protocol $P$ achieves forward secrecy if  $Adv_{P}^{KE-fs}(t,q_{ex})$ is negligible.
\end{definition}

\section{Security analysis of the proposed GKE protocol}\label{sec5}
We begin this section by proving that after running Protocol \ref{Prot1} all participants obtain the same shared key.

\begin{proposition}\label{ProGKEx}
The Protocol \ref{Prot1} verifies the correctness condition.  
\end{proposition}

\begin{proof}
It is clear that if $sid_{U}^{i}=sid_{U^{'}}^{j}$, $pid_{U}^{i}=pid_{U^{'}}^{j}$, and $acc_{U}^{i}=acc_{U^{'}}^{j}=TRUE$ then $U_{i}$ and $U_{j}^{'}$ have received the same $\left\{v_{\alpha}\right\}, \left\{w_{\alpha}\right\}, \left\{Z_{\alpha}\right\}$. Now, let us consider $\mathcal{U}_{i}\in\mathcal{U}$. Then each player $\mathcal{U}_{i}$ computes its $sk_{U}^{i}$, and we will prove by induction on $j$ that $sk_{U}^{i}=sk_{U}^{j}=\phi(h_{1}\odot h_{n},g)\cdot\phi(h_{2}\odot h_{1},g)\cdots\phi(h_{n}\odot h_{n-1},g)$ for any $i,j\in\left\{1, \ldots, n\right\}$, so in particular $sk_{U}^{i}=sk_{U}^{j}\neq NULL$. \\
To begin with, we have that $sk_{U}^{1}=\phi(h_{1}\odot h_{n},g)\cdot\phi(h_{2}\odot h_{1},g)\cdots\phi(h_{n}\odot h_{n-1},g)$. Let us suppose that 
\begin{align*}
sk_{U}^{j} & =A_{\sigma^{j-1}(1)}^{j}\cdot A_{\sigma^{j-1}(2)}^{j},A_{\sigma^{j-1}(3)}^{j},\ldots,A_{\sigma^{j-1}(n)}^{j}, \\
					 & =\phi(h_{1}\odot h_{n},g)\cdot\phi(h_{2}\odot h_{1},g)\cdots\phi(h_{n}\odot h_{n-1},g).
\end{align*}
Now, for $k=j+1$ the following relations are satisfied:
\begin{align*}
A_{1}^{k} & =A_{2}^{j},  \\
A_{2}^{k} & =A_{3}^{j},  \\
\vdots & =\vdots  \\
A_{n}^{k} & =A_{1}^{j}.  
\end{align*}
As a consequence, 
\begin{align*}
sk_{U}^{k} & =A_{\sigma^{k-1}(1)}^{k}\cdot A_{\sigma^{k-1}(2)}^{k},A_{\sigma^{k-1}(3)}^{k},\ldots,A_{\sigma^{k-1}(n)}^{k}, \\
           & =A_{\sigma^{j-1}\circ\sigma(1)}^{K}\cdot A_{\sigma^{j-1}\circ\sigma(2)}^{k},A_{\sigma^{j-1}\circ\sigma(3)}^{k},\ldots,A_{\sigma^{j-1}\circ\sigma(n)}^{k}, \\
					 & =A_{\sigma^{j-1}(1)}^{j}\cdot A_{\sigma^{j-1}(2)}^{j},A_{\sigma^{j-1}(3)}^{j},\ldots,A_{\sigma^{j-1}(n)}^{j},
\end{align*}
so $sk_{U}^{j+1}=sk_{U}^{j}=\phi(h_{1}\odot h_{n},g)\cdot\phi(h_{2}\odot h_{1},g)\cdots\phi(h_{n}\odot h_{n-1},g)$, and $sk_{U}^{i}=sk_{U}^{j}\neq NULL$ for any $i,j\in\left\{1,\ldots,n\right\}$.
\end{proof}

The security of the Protocol \ref{Prot1} is based on the decisional Diffie-Hellman for group actions (DDH-GA). Given $(G,\cdot)$, $(H,\odot)$ two finite groups, $\phi: H\times G\rightarrow G$ a finite group action as in Definition \ref{DefGroupAction}, and an element $g\in G$, the decisional Diffie-Hellman problem for a group action is to distinguish with non-negligible advantage between
the distributions ($\phi(x,g),\phi(y,g),\phi(y\odot x,g),\phi(x\odot y,g)$) and ($\phi(x,g),\phi(y,g),\phi(z,g),\phi(r,g)$), where $z,r\in H\setminus\left\{y\odot xH_{g},x\odot yH_{g}\right\}$ are chosen at random.\\
More formally, we define $Adv_{G}^{ddh-ga}(t)$ as the maximum value, over all distinguishing algorithms $B$ running in time at most $t$, of:
\begin{multline*}
\left|Pr\left[x,y\leftarrow H: B(\phi(x,g),\phi(y,g),\phi(y\odot x,g),\phi(x\odot y,g)=1\right]-\right. \\ \left. Pr\left[x,y\leftarrow H; z,r\leftarrow H\setminus\left\{y\odot xH_{g},x\odot yH_{g}\right\}:B(\phi(x,g),\phi(y,g),\phi(z,g),\phi(r,g))=1\right]\right|.
\end{multline*}
We say that a finite action $\phi: H\times G\rightarrow G$ satisfies the $DDH-GA$ assumption if $Adv_{G}^{ddh-ga}(t)$ is ``small'' for ``reasonable'' values of $t$.

Next, we prove that the Protocol \ref{Prot1} achieves forward security in the security model of the previous section. Our proof, although it is based on the proof of \cite[Theorem 2.]{KatzMoti07}, differs from the original one due to the need of adapting both the presented distributions and the factor measuring the statistical closeness to the case of non-abelian finite group actions.

\begin{theorem}
Given a finite action $\phi: H\times G\rightarrow G$ that satisfies the $DDH-GA$ assumption, the protocol proposed above \ref{Prot1} is a secure conference key distribution protocol achieving forward secrecy, that is,
\begin{equation*}
Adv_{P}^{KE-fs}(t,q_{ex})\leq 4\cdot Adv_{G}^{ddh-ga}(t^{'})+\frac{2\cdot q_{ex}\cdot \left|H_{g}\right|}{\left|H\right|},
\end{equation*}
where $t^{'}=t+\mathcal{O}(\left|\mathcal{P}\right|\cdot q_{ex}\cdot t_{\phi})$, $\mathcal{P}$ denotes the size of the set of potential participants and $t_{\phi}$ is the time to perform $\phi(h,g)$.
\end{theorem}

\begin{proof}
Assume an adversary $\mathcal{A}$ making a single query $Execute(\mathcal{U}_{1}, \ldots, \mathcal{U}_{n})$, where the number of parties $n$ is chosen by the adversary. Let $n=3s+k$, where $k\in\left\{3,4,5\right\}$ and $s\geq 0$ is an integer. Considering an execution of the Protocol \ref{Prot1}, we define $\Gamma_{i+1,i}=\phi(h_{i+1}\odot h_{i},g)$ and note that
\begin{equation}
Z_{i}\equiv X_{i}^{-1}\cdot Y_{i}=\phi(h_{i}\odot h_{i-1},g)^{-1}\cdot\phi(h_{i+1}\odot h_{i},g)=\Gamma_{i,i-1}^{-1}\cdot\Gamma_{i+1,i}.
\end{equation}
Moreover, with the previous notation, the common session key is equal to
\begin{align*}
sk_{U}^{1} & =X_{1}\cdot(X_{1}\cdot Z_{1})\cdot(X_{1}\cdot Z_{1}\cdot Z_{2})\cdots(X_{1}\cdot Z_{1}\cdots Z_{n-2}\cdot Z_{n-1}), \\
           & =\Gamma_{1,n}\cdot(\Gamma_{1,n}\cdot Z_{1})\cdot(\Gamma_{1,n}\cdot Z_{1}\cdot Z_{2})\cdots(\Gamma_{1,n}\cdot Z_{1}\cdots Z_{n-2}\cdot Z_{n-3}),
\end{align*}
As a consequence, in a real execution of the protocol the distribution of the transcript $T$ and the resulting session key $sk$ is given by the following distribution, denoted as $Real$:
\begin{equation*}
Real\equiv\begin{cases} 
h_{1}, \ldots, h_{n}\leftarrow H; & \\
v_{1}=\phi(h_{1},g), v_{2}=\phi(h_{2},g), \ldots, v_{n}=\phi(h_{n},g); & \\
\Gamma_{1,n}=\phi(h_{1}\odot h_{n},g), \Gamma_{2,1}=\phi(h_{2}\odot h_{1},g), \ldots, \Gamma_{n,n-1}=\phi(h_{n}\odot h_{n-1},g); & \\
w_{1}=\phi(c_{n},\Gamma_{1,n}), w_{2}=\phi(c_{1},\Gamma_{2,1}), \ldots, w_{n}=\phi(c_{n-1},\Gamma_{n.n-1}); & \\
Z_{1}=\Gamma_{1,n}^{-1}\cdot\Gamma_{2,1}, Z_{2}=\Gamma_{2,1}^{-1}\cdot\Gamma_{3,2}, \ldots, Z_{n}=\Gamma_{n,n-1}^{-1}\cdot\Gamma_{1,n}; & \\
T=(v_{1},\ldots,v_{n},w_{1},\ldots,w_{n},Z_{1},\ldots,Z_{n}); & \\ 
sk=\Gamma_{1,n}\cdot(\Gamma_{1,n}\cdot Z_{1})\cdots(\Gamma_{1,n}\cdot Z_{1}\cdots Z_{n-2}\cdot Z_{n-3}); & \\
: (T,sk) &
\end{cases}
\end{equation*}
We next define a distribution $Fake^{'}$ in the following way: random elements $\left\{h_{i}\right\}$ as in the case of $Real$. The $\left\{v_{i}\right\}$ are also computed exactly as in $Real$. However, the values $\Gamma_{2,1}, \Gamma_{3,2}, \Gamma_{4,3}$, as well as every value $\Gamma_{j+1,j}$ with $j\geq 6$ a multiple of $3$, are now chosen uniformly at random from $G$ (recall $n=3s+5$ and $s\geq 1$).
\begin{equation*}
Fake^{'}\equiv\begin{cases} 
h_{1}, \ldots, h_{n}\leftarrow H; & \\
v_{1}=\phi(h_{1},g), v_{2}=\phi(h_{2},g), \ldots, v_{n}=\phi(h_{n},g); & \\
\Gamma_{2,1}, \Gamma_{3,2},\Gamma_{4,3}\leftarrow G; \Gamma_{5,4}=\phi(h_{5}\odot h_{4},g);  & \\
\text{for}\enspace i=1\enspace\text{to}\enspace s: & \\
\text{let}\enspace j=3i+3 &  \\
\Gamma_{j,j-1}=\phi(h_{j}\odot h_{j-1},g), \Gamma_{j+1,j}\leftarrow G, \Gamma_{j+2,j+1}=\phi(h_{j+2}\odot h_{j+1},g); & \\
\Gamma_{1,n}=\phi(h_{1}\cdot h_{n},g); & \\
w_{1}=\phi(c_{n},\Gamma_{1,n}), w_{2}=\phi(c_{1},\Gamma_{2,1}), \ldots, w_{n}=\phi(c_{n-1},\Gamma_{n.n-1}); & \\
Z_{1}=\Gamma_{1,n}^{-1}\cdot\Gamma_{2,1}, Z_{2}=\Gamma_{2,1}^{-1}\cdot\Gamma_{3,2}, \ldots, Z_{n}=\Gamma_{n,n-1}^{-1}\cdot\Gamma_{1,n}; & \\
T=(v_{1},\ldots,v_{n},w_{1},\ldots,w_{n},Z_{1},\ldots,Z_{n}); & \\
sk=\Gamma_{1,n}\cdot(\Gamma_{1,n}\cdot Z_{1})\cdots(\Gamma_{1,n}\cdot Z_{1}\cdots Z_{n-2}\cdot Z_{n-3}). & \\
: (T,sk) &
\end{cases}
\end{equation*}
\begin{claim}\label{Claim1}
For any algorithm $\mathcal{A}$ running in time $t$ we have:
\begin{equation*}
\left|Pr\left[(T,sk)\leftarrow Real: \mathcal{A}(T,sk)=1\right]-Pr\left[(T,sk)\leftarrow Fake^{'}: \mathcal{A}(T,sk)=1\right]\right|\leq\epsilon(t^{''})+\frac{\left|H_{g}\right|}{\left|H\right|},
\end{equation*}
where $\epsilon(\cdot)\equiv Adv_{G}^{ddh-ga}(\cdot)$.
\end{claim}

\begin{proof}
Given an algorithm $\mathcal{A}$, consider the following algorithm $D$ which takes as input a tuple ($x,y,z,r$)$\in H^{4}$: $D$ generates ($T,sk$) according to distribution $Dist^{'}$, runs $\mathcal{A}$($T,sk$), and outputs whatever $\mathcal{A}$ outputs. Distribution $Dist^{'}$ is defined as follows:
\begin{equation*}
Dist^{'}\equiv\begin{cases} 
\beta_{0},\beta_{0}^{'},h_{0},\left\{\beta_{i},\gamma_{i},h_{i}\right\}_{i=1}^{s}\leftarrow H; & \\
v_{1}=\phi(y\odot\beta_{0},g), v_{2}=\phi(x,g), v_{3}=\phi(y,g), v_{4}=\phi(\beta_{0}^{'}\odot x,g), v_{5}=\phi(h_{0},g); & \\
\Gamma_{2,1}=\phi(r\odot \beta_{0},g), \Gamma_{3,2}=\phi(z,g),\Gamma_{4,3}=\phi(\beta_{0}^{'}\odot r,g), \Gamma_{5,4}=\phi(h_{0},v_{4});  & \\
\text{for}\enspace i=1\enspace\text{to}\enspace s: & \\
\text{let}\enspace j=3i+3 &  \\
v_{j}=\phi(x\odot\gamma_{i},g), v_{j+1}=\phi(\beta_{i}\odot y,g), v_{j+2}=\phi(h_{i},g); & \\
\Gamma_{j,j-1}=\phi(x\odot\gamma_{i}\odot h_{i-1},g), \Gamma_{j+1,j}=\phi(\beta_{i}\odot z\odot\gamma_{i},g), \Gamma_{j+2,j+1}=\phi(h_{i},v_{j+1}); & \\
\Gamma_{1,n}=\phi(h_{s},v_{1}); & \\
w_{1}=\phi(c_{n},\Gamma_{1,n}), w_{2}=\phi(c_{1},\Gamma_{2,1}), \ldots, w_{n}=\phi(c_{n-1},\Gamma_{n,n-1}); & \\
Z_{1}=\Gamma_{1,n}^{-1}\cdot\Gamma_{2,1}, Z_{2}=\Gamma_{2,1}^{-1}\cdot\Gamma_{3,2}, \ldots, Z_{n}=\Gamma_{n,n-1}^{-1}\cdot\Gamma_{1,n}; & \\
T=(v_{1},\ldots,v_{n},w_{1},\ldots,w_{n},Z_{1},\ldots,Z_{n}); & \\ 
sk=\Gamma_{1,n}\cdot(\Gamma_{1,n}\cdot Z_{1})\cdots(\Gamma_{1,n}\cdot Z_{1}\cdots Z_{n-2}\cdot Z_{n-3}). & \\
: (T,sk). &
\end{cases}
\end{equation*}
We first examine the above distribution when ($x,y,z,r$) is chosen uniformly at random from the set of tuples ($x,y,z=y\odot x,r=x\odot y$); we refer to the resulting distribution as $Dist^{'}_{ddh-ga}$:
\begin{equation*}
Dist_{ddh-ga}^{'}\equiv\begin{cases} 
\beta_{0},\beta_{0}^{'},h_{0},\left\{\beta_{i},\gamma_{i},h_{i}\right\}_{i=1}^{s}\leftarrow H; & \\
v_{1}=\phi(y\odot\beta_{0},g), v_{2}=\phi(x,g), v_{3}=\phi(y,g), v_{4}=\phi(\beta_{0}^{'}\odot x,g), v_{5}=\phi(h_{0},g); & \\
\Gamma_{2,1}=\phi(x\odot y\odot \beta_{0},g), \Gamma_{3,2}=\phi(y\odot x,g),\Gamma_{4,3}=\phi(\beta_{0}^{'}\odot x\odot y,g), \Gamma_{5,4}=\phi(h_{0},v_{4});  & \\
\text{for}\enspace i=1\enspace\text{to}\enspace s: & \\
\text{let}\enspace j=3i+3 & \\
v_{j}=\phi(x\odot\gamma_{i},g), v_{j+1}=\phi(\beta_{i}\odot y,g), v_{j+2}=\phi(h_{i},g); & \\
\Gamma_{j,j-1}=\phi(x\odot\gamma_{i}\odot h_{i-1},g), \Gamma_{j+1,j}=\phi(\beta_{i}\odot y\odot x\odot \gamma_{i},g), \Gamma_{j+2,j+1}=\phi(h_{i},v_{j+1}); & \\
\Gamma_{1,n}=\phi(h_{s},v_{1}); & \\
w_{1}=\phi(c_{n},\Gamma_{1,n}), w_{2}=\phi(c_{1},\Gamma_{2,1}), \ldots, w_{n}=\phi(c_{n-1},\Gamma_{n,n-1}); & \\
Z_{1}=\Gamma_{1,n}^{-1}\cdot\Gamma_{2,1}, Z_{2}=\Gamma_{2,1}^{-1}\cdot\Gamma_{3,2}, \ldots, Z_{n}=\Gamma_{n,n-1}^{-1}\cdot\Gamma_{1,n}; & \\
T=(v_{1},\ldots,v_{n},w_{1},\ldots,w_{n},Z_{1},\ldots,Z_{n}); & \\
sk=\Gamma_{1,n}\cdot(\Gamma_{1,n}\cdot Z_{1})\cdots(\Gamma_{1,n}\cdot Z_{1}\cdots Z_{n-2}\cdot Z_{n-3}). & \\
: (T,sk),
\end{cases}
\end{equation*}
where all we have done is substitute $z$ and $r$ into the definition of $Dist^{'}$. We claim that $Dist^{'}_{ddh-ga}$ is identical to $Real$. To see this, first notice that $\left\{Z_{i}\right\}$, $T$, and $sk$ are computed identically in both, so it suffices to look at the distribution of the $v's$, the $w's$, and the $\Gamma's$. It is not hard to see that in $Dist^{'}_{ddh-ga}$ the $\left\{v_{k}\right\}_{k=1}^{n}$ and the $\left\{w_{k}\right\}_{k=1}^{n}$ are uniformly and independently distributed in $G$, exactly as in $Real$. It remains to show that for all $k\in\left\{1,\ldots,n\right\}$, the tuple $\left\{v_{k},v_{k+1},\Gamma_{k+1,k}\right\}$ in $Dist^{'}_{ddh-ga}$ is a Diffie-Hellman group action tuple as in $Real$, what can be verified by computation. We conclude that:
\begin{equation}\label{Eq3}
Pr\left[x,y\leftarrow H: D(x,y,y\odot x,x\odot y)=1\right]=Pr\left[(T,sk)\leftarrow Real: \mathcal{A}(T,sk)=1\right].
\end{equation}
We now examine distribution $Dist^{'}$ in case ($x,y,z,r$) is chosen uniformly from the space of random tuples, refereeing to it as $Dist^{'}_{rand}$:
\begin{equation*}
Dist^{'}_{rand}\equiv\begin{cases} 
\beta_{0},\beta_{0}^{'},h_{0},\left\{\beta_{i},\gamma_{i},h_{i}\right\}_{i=1}^{s}\leftarrow H; & \\
v_{1}=\phi(y\odot\beta_{0},g), v_{2}=\phi(x,g), v_{3}=\phi(y,g), v_{4}=\phi(\beta_{0}^{'}\odot x,g), v_{5}=\phi(h_{0},g); & \\
\Gamma_{2,1}=\phi(r\odot \beta_{0},g), \Gamma_{3,2}=\phi(z,g),\Gamma_{4,3}=\phi(\beta_{0}^{'}\odot r,g), \Gamma_{5,4}=\phi(h_{0},v_{4});  & \\
\text{for}\enspace i=1\enspace\text{to}\enspace s: & \\
\text{let}\enspace j=3i+3 &  \\
v_{j}=\phi(x\odot\gamma_{i},g), v_{j+1}=\phi(\beta_{i}\odot y,g), v_{j+2}=\phi(h_{i},g); & \\
\Gamma_{j,j-1}=\phi(x\odot\gamma_{i}\odot h_{i-1},g), \Gamma_{j+1,j}=\phi(\beta_{i}\odot z\odot\gamma_{i},g), \Gamma_{j+2,j+1}=\phi(h_{i},v_{j+1}); & \\
\Gamma_{1,n}=\phi(h_{s},v_{1}); & \\
w_{1}=\phi(c_{n},\Gamma_{1,n}), w_{2}=\phi(c_{1},\Gamma_{2,1}), \ldots, w_{n}=\phi(c_{n-1},\Gamma_{n,n-1}); & \\
Z_{1}=\Gamma_{1,n}^{-1}\cdot\Gamma_{2,1}, Z_{2}=\Gamma_{2,1}^{-1}\cdot\Gamma_{3,2}, \ldots, Z_{n}=\Gamma_{n,n-1}^{-1}\cdot\Gamma_{1,n}; & \\
T=(v_{1},\ldots,v_{n},w_{1},\ldots,w_{n},Z_{1},\ldots,Z_{n}); & \\
sk=\Gamma_{1,n}\cdot(\Gamma_{1,n}\cdot Z_{1})\cdots(\Gamma_{1,n}\cdot Z_{1}\cdots Z_{n-2}\cdot Z_{n-3}). & \\
: (T,sk)
\end{cases}
\end{equation*}
We claim that $Dist^{'}_{rand}$ is statistically close to $Fake^{'}$. As before, we only need to examine the distribution of the $v's$, the $w's$ and the $\Gamma's$. It is again not hard to see that in $Dist^{'}_{rand}$ the $\left\{v_{k}\right\}$ and the $\left\{w_{k}\right\}$ are uniformly and independently distributed in $G$, as in $Fake^{'}$. By inspection, we see also that the $\Gamma_{5,4}$, $\left\{\Gamma_{j,j-1}, \Gamma_{j+2,j+1}\right\}_{j=3i+3; i=1,\ldots,s}$, and $\Gamma_{1,n}$ are distributed identically in $Dist^{'}_{rand}$ and $Fake^{'}$. It only remains to show that in $Dist^{'}_{rand}$, the distribution of the remaining $\Gamma's$ is statistically close to the uniform one. We take $\Gamma_{j+1,j}$(for some $i\in\left\{1,\ldots,s\right\}$ and $j=3i+3$) as a representative example. Since $\beta_{i}$ is used only in computing $v_{j+1}$ and $\Gamma_{j+1,j}$, the joint distribution of $v_{j+1}, \Gamma_{j+1,j}$ is given by:
\begin{align*}
v_{j+1} & =\phi(\beta_{i}\odot y,g), \\
\Gamma_{j+1,j} & =\phi(\beta_{i}\odot z\odot\gamma_{i},g),
\end{align*}
where $\beta_{i},\gamma_{i}$ are chosen uniformly and independently from $H$. Using the fact that $z\notin y\odot xH_{g}$, the above relations are independent, and it follows that $\Gamma_{j+1,j}$ is uniformly distributed in $G$, independent of $v_{j+1}$ and everything else. A similar argument applies for $\Gamma_{2,1}$ and $\Gamma_{4,3}$. The only value left is $\Gamma_{3,2}$. Because of our restriction on the choice of $z$, the value $h\in H$ verifying $\phi(h,g)=\Gamma_{3,2}$ in $Dist^{'}_{rand}$ is distributed uniformly in $H\setminus y\odot xH_{g}$. This is statistically close (within a factor of $\frac{\left|H_{g}\right|}{\left|H\right|}$) to the distribution of $h^{'}\in H$ with $\phi(h^{'},g)=\Gamma_{3,2}$ in $Fake^{'}$. We conclude that:
\begin{multline}\label{Eq4}
Pr\left[x,y\leftarrow H; z,r\leftarrow H\setminus\left\{y\odot xH_{g}, x\odot yH_{g}\right\}: D(x,y,z,r)=1\right]\geq \\ Pr\left[(T,sk)\leftarrow Fake^{'}: \mathcal{A}(T,sk)=1\right]-\frac{\left|H_{g}\right|}{\left|H\right|}.
\end{multline}
The running time of $D$ if $t^{''}$, the running time of $\mathcal{A}$ plus $\mathcal{O}(n)$ computations of $\phi(h,g)$. The claim now follows easily from Equations \ref{Eq3} and $\ref{Eq4}$ and the definition of $\epsilon$.  
\end{proof}
We now introduce one final distribution in which all the $\Gamma's$ are chosen uniformly and independently at random from $G$:
\begin{equation*}
Fake\equiv\begin{cases} 
h_{1}, \ldots, h_{n}\leftarrow H; & \\
v_{1}=\phi(h_{1},g), v_{2}=\phi(h_{2},g), \ldots, v_{n}=\phi(h_{n},g); & \\
\Gamma_{2,1},\ldots,\Gamma_{n,n-1},\Gamma_{1,n}\leftarrow G;  &  \\
w_{1}=\phi(c_{n},\Gamma_{1,n}), w_{2}=\phi(c_{1},\Gamma_{2,1}), \ldots, w_{n}=\phi(c_{n-1},\Gamma_{n.n-1}); & \\
Z_{1}=\Gamma_{1,n}^{-1}\cdot\Gamma_{2,1}, Z_{2}=\Gamma_{2,1}^{-1}\cdot\Gamma_{3,2}, \ldots, Z_{n}=\Gamma_{n,n-1}^{-1}\cdot\Gamma_{1,n}; & \\
T=(v_{1},\ldots,v_{n},w_{1},\ldots,w_{n},Z_{1},\ldots,Z_{n}); & \\
sk=\Gamma_{1,n}\cdot(\Gamma_{1,n}\cdot Z_{1})\cdots(\Gamma_{1,n}\cdot Z_{1}\cdots Z_{n-2}\cdot Z_{n-3}). & \\
: (T,sk) &
\end{cases}
\end{equation*}
\begin{claim}\label{Claim2}
For any algorithm $\mathcal{A}$ running in time $t$ we have:
\begin{equation*}
\left|Pr\left[(T,sk)\leftarrow Fake^{'}: \mathcal{A}(T,sk)=1\right]-Pr\left[(T,sk)\leftarrow Fake: \mathcal{A}(T,sk)=1\right]\right|\leq\epsilon(t^{''}).
\end{equation*}
\end{claim}
\begin{proof}
Given an algorithm $\mathcal{A}$, consider the following algorithm $D$ which takes as input a tuple ($x,y,z,r$)$\in H^{4}$: $D$ generates ($T,sk$) according to distribution $Dist$, runs $\mathcal{A}(T,sk)$, and outputs whatever $\mathcal{A}$ outputs. Distribution $Dist$ is defined as follows:
\begin{equation*}
Dist\equiv\begin{cases} 
h_{1},h_{2},\left\{\beta_{i},\beta_{i}^{'},\gamma_{i}\right\}_{i=0}^{s}\leftarrow H; & \\
v_{1}=\phi(y\odot\beta_{0},g), v_{2}=\phi(h_{1},g), v_{3}=\phi(h_{2},g), v_{4}=\phi(y\odot\beta_{0}^{'},g), v_{5}=\phi(\gamma_{0}\odot x,g); & \\
\Gamma_{2,1}, \Gamma_{3,2},\Gamma_{4,3}\leftarrow G; \Gamma_{5,4}=\phi(\gamma_{0}\odot r\odot\beta_{0}^{'},g);  & \\
\text{for}\enspace i=1\enspace\text{to}\enspace s: & \\
\text{let}\enspace j=3i+3 &  \\
v_{j}=\phi(\beta_{i}\odot y\odot\gamma_{i-1}^{-1},g), v_{j+1}=\phi(y\odot\beta_{i}^{'},g), v_{j+2}=\phi(\gamma_{i}\odot x,g); &  \\
\Gamma_{j,j-1}=\phi(\beta_{i}\odot z,g), \Gamma_{j+1,j}\leftarrow G, \Gamma_{j+2,j+1}=\phi(\gamma_{i}\odot r\odot\beta_{i}^{'},g); & \\
\Gamma_{1,n}=\phi(\gamma_{s}\odot t\odot\beta_{0},g); & \\
w_{1}=\phi(c_{n},\Gamma_{1,n}), w_{2}=\phi(c_{1},\Gamma_{2,1}), \ldots, w_{n}=\phi(c_{n-1},\Gamma_{n,n-1}); & \\
Z_{1}=\Gamma_{1,n}^{-1}\cdot\Gamma_{2,1}, Z_{2}=\Gamma_{2,1}^{-1}\cdot\Gamma_{3,2}, \ldots, Z_{n}=\Gamma_{n,n-1}^{-1}\cdot\Gamma_{1,n}; & \\
T=(v_{1},\ldots,v_{n},w_{1},\ldots,w_{n},Z_{1},\ldots,Z_{n}); & \\
sk=\Gamma_{1,n}\cdot(\Gamma_{1,n}\cdot Z_{1})\cdots(\Gamma_{1,n}\cdot Z_{1}\cdots Z_{n-2}\cdot Z_{n-3}). & \\
: (T,sk) &
\end{cases}
\end{equation*}
The remainder of the proof is very similar to the proof of the previous claim. We first examine the above distribution when ($x,y,z=y\odot x,r=x\odot y$) is chosen uniformly from the set of Diffie-Hellman tuples, and we refer to the resulting distribution as $Dist_{ddh-ga}$:
\begin{equation*}
Dist_{ddh-ga}\equiv\begin{cases} 
h_{1},h_{2},\left\{\beta_{i},\beta_{i}^{'},\gamma_{i}\right\}_{i=0}^{s}\leftarrow H; & \\
v_{1}=\phi(y\odot\beta_{0},g), v_{2}=\phi(h_{1},g), v_{3}=\phi(h_{2},g), v_{4}=\phi(y\odot\beta_{0}^{'},g), v_{5}=\phi(\gamma_{0}\odot x,g); & \\
\Gamma_{2,1}, \Gamma_{3,2},\Gamma_{4,3}\leftarrow G; \Gamma_{5,4}=\phi(\gamma_{0}\odot x\odot y\odot\beta_{0}^{'},g);  & \\
\text{for}\enspace i=1\enspace\text{to}\enspace s: & \\
\text{let}\enspace j=3i+3 &  \\
v_{j}=\phi(\beta_{i}\odot y\odot\gamma_{i-1}^{-1},g), v_{j+1}=\phi(y\odot\beta_{i}^{'},g), v_{j+2}=\phi(\gamma_{i}\odot x,g); &  \\
\Gamma_{j,j-1}=\phi(\beta_{i}\odot y\odot x,g), \Gamma_{j+1,j}\leftarrow G, \Gamma_{j+2,j+1}=\phi(\gamma_{i}\odot x\odot y\odot\beta_{i}^{'},g); & \\
\Gamma_{1,n}=\phi(\gamma_{s}\odot x\odot y\odot\beta_{0},g); & \\
w_{1}=\phi(c_{n},\Gamma_{1,n}), w_{2}=\phi(c_{1},\Gamma_{2,1}), \ldots, w_{n}=\phi(c_{n-1},\Gamma_{n,n-1}); & \\
Z_{1}=\Gamma_{1,n}^{-1}\cdot\Gamma_{2,1}, Z_{2}=\Gamma_{2,1}^{-1}\cdot\Gamma_{3,2}, \ldots, Z_{n}=\Gamma_{n,n-1}^{-1}\cdot\Gamma_{1,n}; & \\
T=(v_{1},\ldots,v_{n},w_{1},\ldots,w_{n},Z_{1},\ldots,Z_{n}); & \\ 
sk=\Gamma_{1,n}\cdot(\Gamma_{1,n}\cdot Z_{1})\cdots(\Gamma_{1,n}\cdot Z_{1}\cdots Z_{n-2}\cdot Z_{n-3}). & \\
: (T,sk) &
\end{cases}
\end{equation*}
We claim that $Dist_{ddh-ga}$ is identical to $Fake^{'}$. As in the proof of the previous claim, we need only focus on the $v's$, $w's$ and the $\Gamma's$. It is again easy to see that in $Dist_{ddh-ga}$ the $\left\{v_{k}\right\}$ and the $\left\{w_{k}\right\}$ are uniformly and independently distributed in $G$, as in $Fake^{'}$. It is also immediate that in both distributions we have that $\Gamma_{2,1}, \Gamma_{3,2}, \Gamma_{4,3}$ and $\left\{\Gamma_{j+1,j}\right\}_{j=3i+3; i=1,\ldots, s}$ are distributed uniformly and independently in $G$. As for the remaining $\Gamma's$, we consider $\Gamma_{j,j-1}$ as a representative example. We claim that in $Dist_{ddh-ga}$ the tuple ($v_{j},v_{j-1},\Gamma_{j,j-1}$) is a Diffie-Hellman group action tuple, as in the case of $Fake^{'}$. Thus, we conclude that:
\begin{equation}\label{Eq5}
Pr\left[x,y\leftarrow H: D(x,y,y\odot x,x\odot y)=1\right]=Pr\left[(T,sk)\leftarrow Fake^{'}: \mathcal{A}(T,sk)=1\right].
\end{equation}
We now examine distribution $Dist$ when ($x,y,z,r$) is chosen uniformly from the set of random tuples, and we refer to the resulting distribution as $Dist_{rand}$:
\begin{equation*}
Dist_{rand}\equiv\begin{cases} 
h_{1},h_{2},\left\{\beta_{i},\beta_{i}^{'},\gamma_{i}\right\}_{i=0}^{s}\leftarrow H; & \\
v_{1}=\phi(y\odot\beta_{0},g), v_{2}=\phi(h_{1},g), v_{3}=\phi(h_{2},g), v_{4}=\phi(y\odot\beta_{0}^{'},g), v_{5}=\phi(\gamma_{0}\odot x,g); & \\
\Gamma_{2,1}, \Gamma_{3,2},\Gamma_{4,3}\leftarrow G; \Gamma_{5,4}=\phi(\gamma_{0}\odot r\odot\beta_{0}^{'},g);  & \\
\text{for}\enspace i=1\enspace\text{to}\enspace s: & \\
\text{let}\enspace j=3i+3 &  \\
v_{j}=\phi(\beta_{i}\odot y\odot\gamma_{i-1}^{-1},g), v_{j+1}=\phi(y\odot\beta_{i}^{'},g), v_{j+2}=\phi(\gamma_{i}\odot x,g); &  \\
\Gamma_{j,j-1}=\phi(\beta_{i}\odot z,g), \Gamma_{j+1,j}\leftarrow G, \Gamma_{j+2,j+1}=\phi(\gamma_{i}\odot r\odot\beta_{i}^{'},g); & \\
\Gamma_{1,n}=\phi(\gamma_{s}\odot t\odot\beta_{0},g); & \\
w_{1}=\phi(c_{n},\Gamma_{1,n}), w_{2}=\phi(c_{1},\Gamma_{2,1}), \ldots, w_{n}=\phi(c_{n-1},\Gamma_{n,n-1}); & \\
Z_{1}=\Gamma_{1,n}^{-1}\cdot\Gamma_{2,1}, Z_{2}=\Gamma_{2,1}^{-1}\cdot\Gamma_{3,2}, \ldots, Z_{n}=\Gamma_{n,n-1}^{-1}\cdot\Gamma_{1,n}; & \\
T=(v_{1},\ldots,v_{n},w_{1},\ldots,w_{n},Z_{1},\ldots,Z_{n}); & \\
sk=\Gamma_{1,n}\cdot(\Gamma_{1,n}\cdot Z_{1})\cdots(\Gamma_{1,n}\cdot Z_{1}\cdots Z_{n-2}\cdot Z_{n-3}). & \\
: (T,sk) & 
\end{cases}
\end{equation*}
We claim that $Dist_{rand}$ is identical to $Fake$. It is easy to verify that in $Dist_{rand}$ the $\left\{v_{k}\right\}$ and the $\left\{w_{k}\right\}$ are uniformly and independently distributed in $G$. We need to show that the $\Gamma's$ are also uniformly distributed, independent of each other and the $v's$, as in the case in $Fake$. This clearly holds for $\Gamma_{2,1}, \Gamma_{3,2}, \Gamma_{4,3}$ and $\left\{\Gamma_{j+1,j}\right\}_{j=3i+3; i=1,\ldots, s}$. For the remaining $\Gamma's$, we take $\Gamma_{5,4}$ as a representative example. As $\beta_{0}^{'}$ is used only in computing $v_{4}$ and $\Gamma_{5,4}$, the joint distribution of $v_{4}$ and $\Gamma_{5,4}$ is given by:
\begin{align*}
v_{4} & =\phi(y\odot\beta_{0}^{'},g), \\
\Gamma_{5,4} & =\phi(\gamma_{0}\odot r\odot\beta_{0}^{'},g),
\end{align*}
where $\beta_{0}^{'},\gamma_{0}$ are uniformly and independently distributed in $H$. Using the fact that $r\notin x\odot yH_{g}$, the above distributions are independent. It follows that $\Gamma_{5,4}$ is uniformly distributed in $G$, independent of $v_{4}$ and everything else. As a result, we conclude that:
\begin{multline}\label{Eq6}
Pr\left[x,y\leftarrow H; z,r\leftarrow H\setminus\left\{y\odot xH_{g},x\odot yH_{g}\right\} :D(x,y,z,r)=1\right]= \\ Pr\left[(T,sk)\leftarrow Fake: \mathcal{A}(T,sk)=1\right].
\end{multline}
The running time of $D$ is $t^{''}$. The claim now follows readily from Equation \ref{Eq5} and $\ref{Eq6}$.
\end{proof}
In experiment, $Fake$, let $s_{i+1,i}\in H$ such that $\phi(s_{i+1,i},g)=\Gamma_{i+1,i}$ for $1\leq i\leq n$. Given $T$, the values $s_{2,1},\ldots,s_{1,n}$ are constrained by the following $n$ relations:
\begin{align*}
Z_{1} & =\phi(s_{1,n},g)^{-1}\cdot\phi(s_{2,1},g), \\
\vdots & =\vdots \\
Z_{n} & =\phi(s_{n,n-1},g)^{-1}\cdot\phi(s_{1,n},g).																				
\end{align*}
Furthermore, we have that $sk=\phi(s_{1,n},g)\cdot\phi(s_{2,1},g)\cdots\phi(s_{n,n-1},g)$, so we conclude that $sk$ is independent of $T$. This implies that even for a computationally-unbounded adversary $\mathcal{A}$:
\begin{equation*}
Pr\left[(T,sk_{0})\leftarrow Fake; sk_{1}\leftarrow G; b\leftarrow\left\{0,1\right\}: \mathcal{A}(T,sk_{b})=b\right]=\frac{1}{2}.
\end{equation*}
By combining this with Claim \ref{Claim1} and Claim \ref{Claim2}, we obtain that 
\begin{equation*}
Adv_{P}^{KE-fs}(t,1)\leq 4\cdot\epsilon(t^{''})+\frac{2\left|H_{g}\right|}{\left|H\right|}.
\end{equation*}
For the case $q_{ex}>1$, a standard hybrid argument shows that $Adv_{P}^{KE-fs}(t,q_{ex})\leq q_{ex}\cdot Adv_{P}^{KE-fs}(t,1)$ giving the result of the theorem.
\end{proof}

\section{Some group actions related to hard group-theoretic problems}\label{sec6}
In this section, we explore some finite group actions associated to intractable group-theoretic problems and present the corresponding group key obtained by running Protocol \ref{Prot1}.

\begin{example}[Conjugacy Problem]
Let $(H,\cdot)\leq (G,\cdot)$ be a subgroup and $\phi: H\times G\rightarrow G$ the action of $H$ on $G$ defined by $\phi(h_{i},g)=h_{i}^{-1}\cdot g\cdot h_{i}$. \\
This action belongs to a bigger family of actions that are introduced in the following definition.
\begin{definition}\cite[Ch.3, p.68]{Isaacs08}
Given $(G,\cdot)$ finite group and $(H,\cdot)$ a subgroup, we say that $H$ acts via automorphisms on $G$ if $H$ acts on $G$ as a set, and in addition, $\phi(h,g_{1}\cdot g_{2})=\phi(h,g_{1})\cdot \phi(h,g_{2})$ for all $g_{1},g_{2}\in G$ and $h\in H$.
\end{definition}
Note that if $H$ acts via automorphisms on $G$ then we have that $\phi(h,e_{G})=e_{G}$, so $\phi(h,g)^{-1}=\phi(h,g^{-1})$ for all $g\in G$ and $h\in H$.

The Conjugacy Decisional Problem in a given group $(G,\cdot)$, consist of given the elements $a,b\in G$ determine whether the equality $b=x^{-1}\cdot a \cdot x$ holds for some $x\in G$, and the associated Conjugacy Search Problem to find such $x\in G$. The conjugacy search problem is of interest in complexity theory. Indeed, if you know that $b$ is conjugate to $a$, you can just go over words of the form $x^{-1}\cdot a \cdot x$ and
compare them to $b$ one at a time, until you get a match. This straightforward algorithm is at least exponential-time in the length of $b$, and therefore is considered infeasible for practical purposes.
Thus, if no other algorithm is known for the conjugacy search problem in a group $G$, it is not unreasonable to claim that $x\rightarrow x^{-1}\cdot a \cdot x$ is a one-way function and try to build a (public-key) cryptographic protocol on that. The computational difficulty of this problem in appropriately selected
groups has enabled it to be used in various public-key protocols as the Anshel-Anshel-Goldfeld (AAG) key exchange protocol that is based on the
conjugacy search problem on non-commutative groups \cite{Anshel99}. However, it is now well-known that the conjugacy search problem is unlikely to provide sufficient security level if the platform is a braid group as it was originally proposed (see e.g. \cite{Myasnikov05}, \cite{Myasnikov06}). In fact, electing a suitable platform group for the above protocol is a very non-trivial matter; some requirements on such a group were put forward in \cite{Myasnikov11}: i) the group has to be well known, ii) the word problem in $G$ should have a fast (linear- or quadratic-time) solution by a deterministic algorithm, iii) the conjugacy search problem should not have a subexponential-time solution by a deterministic algorithm, and iv) $G$ should be a group of super-polynomial (i.e., exponential or ``intermediate'') growth.

%We note that there is no known polynomial time algorithm for solving the conjugacy search problem in an arbitrary group with small cancellation conditions $C(4)$ and $T(4)$, but not $C^{'}(\frac{1}{6})$ in order to try to avoid hyperbolic groups (all finitely presented $C^{'}(\frac{1}{6})$ groups are hyperbolic), where the conjugacy search problem can be solved very quickly.

In this particular case, we have that after running Protocol \ref{Prot1}
\begin{equation*}
sk_{U}^{i}=h_{1}^{-1}\cdot h_{n}^{-1}\cdot g\cdot h_{n}\cdot h_{1}\cdots h_{n}^{-1}\cdot h_{n-1}^{-1}\cdot g\cdot h_{n-1}\cdot h_{n}.
\end{equation*}
\end{example}

\begin{example}[Twisted Conjugacy Problem]
Let $(H,\cdot)\leq (G,\cdot)$ be a subgroup, $\varphi$ an endomorphism defined on $G$, and $\phi: H\times G\rightarrow G$ the action of $H$ on $G$ defined by $\phi(h_{i},g)=h_{i}^{-1}\cdot g\cdot\varphi(h_{i})$. 

Given an endomorphism $\varphi: G\rightarrow G$, and the elements $a,b\in G$, the Twisted Conjugacy Decisional Problem consist of determining whether the equality $b=x^{-1}\cdot a\cdot\varphi(x)$ holds for some $x\in G$, and the associated Twisted Conjugacy Search Problem to find such $x\in G$. This problem is a generalization of the conjugacy problem (note that when the endomorphism $\varphi=id$ is taken, then the conjugacy problem is obtained) and it has been used in an authentication protocol in \cite{Shpilrain08} but this was successfully attacked in \cite{Grassl09}. Moreover, a key exchange protocol was proposed using the Twisted Conjugacy Search problem implemented in the discrete Heisenberg group in \cite{Isaiyarasi12}. 

The shared key resulting after running Protocol \ref{Prot1} is
\begin{equation*}
sk_{U}^{i}=h_{1}^{-1}\cdot h_{n}^{-1}\cdot g\cdot\varphi(h_{n}\cdot h_{1})\cdots h_{n}^{-1}\cdot h_{n-1}^{-1}\cdot g\cdot\varphi(h_{n-1}\cdot h_{n}).
\end{equation*}
\end{example}

\begin{example}[Double Coset Membership Problem]
Let $(H,\cdot)\leq (G,\cdot)$ and $(J,\cdot)$ be two non-solvable subgroups, and $\psi: (H,J)\times G\rightarrow G$ the double action of ($H,J$) on $G$ defined by $\psi((h_{i},j_{i}),g)=\phi(h_{i},\gamma(j_{i},g))=h_{i}\cdot g\cdot j_{i}$. \\

Given two subgroups $H,J$ and two elements $a,b\in G$ the Double Coset Membership Decisional Problem consist of deciding whether the double coset $HaJ$ contains the
element $b$, and the associated Double Coset Membership Search Problem to find elements $h\in H$ and $j\in J$such that $b=haj$. This problem was used in a key establishment protocol in \cite{Shpilrain06}. The Triple Decomposition Problem, that is an slight modification of the`previous one, was used as well in order to get a shared secret key in \cite{Isaiyarasi12}. In the double coset problem by considering $a=e_{G}$, you get the factorization problem. Although Fenner and Zhang proved in \cite[Corollary 2.]{Fenner05} that the Double Coset Membership over solvable groups can be solved within error $\epsilon$ by a quantum algorithm that runs in time polynomial in $m+log(1//epsilon)$, where $m$ is the size of the input, provided one of the underlying groups is smoothly solvable. However, there is not an analogue algorithm when the two subgroups $H$ and $J$ are non-solvable.

In this particular case we have that the shared key obtained after running Protocol \ref{Prot1} is
\begin{equation*}
sk_{U}^{i}=h_{1}\cdot h_{n}\cdot g\cdot j_{n}\cdot j_{1}\cdots h_{n}\cdot h_{n-1}\cdot g\cdot j_{n-1}\cdot j_{n}.
\end{equation*}
\end{example}

\section{Some final remarks}

\begin{remark}
Note that by considering a cyclic subgroup, the AAG key exchange protocol could be used to establish the shared secret key $c_{i}$ between $U_{i}$ and $U_{i+1}$ in $Round 1$ of Protocol \ref{Prot1}.
\end{remark}

\begin{remark}
If $G$ is quasi-commutative, that is, for every $a, b\in G$, there is a positive integer $r$ such that $ab=b^{r}a=ba^{r}$ (see e.g \cite{Nagy01}), then we can rid the protocol of Rounds $1$ and $3$.
\end{remark}

\begin{remark}
There are some reasons to consider the family of groups with small cancellation conditions $C(4)$ and $T(4)$, and $C^{'}(\lambda)$ with $\lambda\neq\frac{1}{6}$ (see, for example, the monograph \cite{Lyndon01} for an intensive treatment of small cancellation theory) as platform groups for the actions mentioned in Section \ref{sec6}. To begin with there is no known polynomial time algorithm for solving the conjugacy search problem in an arbitrary group of type $C(4)$ and $T(4)$, but not $C^{'}(\frac{1}{6})$ (all finitely presented $C^{'}(\frac{1}{6})$ groups are hyperbolic so the conjugacy search problem can be solved very quickly). Moreover, in \cite{Ger91} Gersten and Short proved that all groups which satisfy the algebraic $C(4)$ and $T(4)$ conditions are automatic. This is an important fact, as automatic groups provide a method for carrying out practical computations such as the growth function of the group or the reduction of words to normal form rapidly, being all these desirable properties when selecting platform groups (see \cite{Myasnikov11}). Finally, the structure of its subgroups has been well studied (see \cite{Mosalamy86} and \cite{duda23}).
\end{remark}

%%%%%%%%%%%%%%%%%%%%%%%%%%%%%%%%%%%%%%%%%%%%%%%%%%%%%%%%%%%%%%%%%%%%%%%

\bigskip

\bibliographystyle{plain}
\bibliography{BDNC}% common bib file
%% if required, the content of .bbl file can be included here once bbl is generated
%%\input sn-article.bbl

\end{document}